\def\year{2021}\relax
\documentclass[letterpaper]{article} 

\usepackage{times} 
\usepackage{helvet} 
\usepackage{courier} 
\usepackage[hyphens]{url} 
\usepackage{graphicx} 
\usepackage{graphicx} 
\usepackage{caption} 
\usepackage[switch]{lineno}
\usepackage[numbers]{natbib}
\usepackage{authblk}
\usepackage{geometry}
\usepackage{xcolor}
\usepackage{amsmath}
\usepackage{amsthm}
\usepackage{amsfonts}

\graphicspath{{./images/}}
\newtheorem{theorem}{Theorem}

\newtheorem{lemma}[theorem]{Lemma}
\newtheorem{observation}[theorem]{Observation}

\newcommand{\propXt}{\text{PROPm}}
\newcommand{\spx}{\text{CP}_i}
\newcommand{\spxi}[3]{\text{CP}_{#1}(#2,#3)}

\begin{document}
\title{Achieving Proportionality up to the Maximin Item\\ with Indivisible Goods}
\author[a]{Artem Baklanov} \author[b]{Pranav Garimidi} \author[c]{Vasilis Gkatzelis} \author[c]{Daniel Schoepflin}
\affil[a]{Higher School of Economics, St. Petersburg}
\affil[b]{Conestoga High School}
\affil[c]{Drexel University}
\date{}

\maketitle

\begin{abstract}
We study the problem of fairly allocating indivisible goods and focus on the classic fairness notion of proportionality. The indivisibility of the goods is long known to pose highly non-trivial obstacles to achieving fairness, and a very vibrant line of research has aimed to circumvent them using appropriate notions of approximate fairness. Recent work has established that even approximate versions of proportionality (PROPx) may be impossible to achieve even for small instances, while the best known achievable approximations (PROP1) are much weaker. We introduce the notion of \emph{proportionality up to the maximin item} (\propXt) and show how to reach an allocation satisfying this notion for any instance involving up to five agents with additive valuations. {\propXt} provides a well-motivated middle-ground between PROP1 and PROPx, while also capturing some elements of the well-studied maximin share (MMS) benchmark: another relaxation of proportionality that has attracted a lot of attention.
\end{abstract}

\section{Introduction}

We consider the well-studied problem of fairly distributing a set of scarce resources among a group of $n$ agents. This problem is at the heart of the long literature on fair division, initiated by \citet{Stern}, which has recently received renewed interest, partly due to the proliferation of automated resource allocation processes. To reach a fair outcome, such processes need to take into consideration the preferences of the agents, i.e., how much each agent values each of the resources. The most common modelling assumption regarding these preferences is that they are \emph{additive}: each agent $i$ has a value $v_{ij}\geq 0$ for each resource $j$, and her value for a set $S$ of resources is $v_i(S)=\sum_{j\in S}v_{ij}$. But, what would constitute a ``fair'' outcome given such preferences?


One of the classic notions of fairness is \emph{proportionality}. An outcome satisfies proportionality if the value of every agent for the resources that were allocated to them is at least a $1/n$ fraction of her total value for all of the resources. For the case of additive valuations, if $M$ is the set of all the resources, then every agent $i$ should receive a value of at least $\frac{1}{n}\sum_{j\in M} v_{ij}$. This captures fairness in a very intuitive way: since there are $n$ agents in total, if they were to somehow divide the total value equally among them, then each of them should be receiving at least a $1/n$ fraction of it; in fact, they could potentially all receive more than that if they each value different resources. However, it is well-known that achieving proportionality may be impossible when the resources are \emph{indivisible}, i.e., cannot be divided into smaller parts and shared among the agents. This can be readily verified with the simple example involving only a single indivisible resource and at least two agents competing for it. In this case, whoever is allocated that resource will receive all of her value but all other agents will receive none of it, violating proportionality.

In light of this impossibility to achieve proportionality in the presence of indivisible resources, the literature has turned to relaxations of this property. A natural candidate would be a multiplicative approximation of proportionality, aiming to guarantee that every agent receives at least a $\lambda/n$ fraction of their total value, for some $\lambda \in [0, 1]$. However, the single resource example provided above directly implies that no $\lambda>0$ is small enough to guarantee the existence of such an approximation. As a result, research has instead considered additive approximations, leading to two interesting notions: PROP1 and PROPx. These relaxations allow the value of each agent $i$ to be less than a $1/n$ fraction of her total value but by no more than some additive difference $d_i$. For the case of PROP1, $d_i$ corresponds to the \textit{maximum} value of agent $i$ over all the items that were allocated to some other agent~\cite{CFS17}. For the case of PROPx, $d_i$ corresponds to the \textit{minimum} value of agent $i$ over all the items that were allocated to some other agent~\cite{AMS2020}. On one hand, PROP1 is a bit too forgiving, and is known to be easy to satisfy, while on the other PROPx is too demanding and is not guaranteed to exist even for instances with just three agents.


In a parallel line of work, an alternative relaxation that has received a lot of attention is the \textit{maximin share} (MMS)~\cite{Budish10}. According to this notion, every agent's ``fair share'' is defined as the value that the agent could secure if she could choose any feasible partition of the resources into $n$ bundles, but was then allocated her least preferred bundle among them. It is not hard to verify that this benchmark is weakly smaller than the one imposed by proportionality, yet prior work has shown that this too may be impossible to achieve, even for instances with just three agents.

In this paper, we propose $\propXt$, a new notion that provides a middle-ground between PROP1 and PROPx, while also capturing the ``maximin flavor'' of the MMS benchmark, and we prove that there always exists an allocation satisfying $\propXt$ for any instance involving up to five agents.





\section{Additional Related Work}


The \emph{proportionality up to the most valued item} (PROP1) notion is a relaxation of proportionality that was introduced by \citet{CFS17}, who observed that there always exists a Pareto optimal allocation that satisfies PROP1. \citet{ACIW19} later extended  this notion to settings where the objects being allocated are chores, i.e., the valuations are negative, and very recently \citet{AMS2020} provided a strongly polynomial time algorithm for computing allocations that are Pareto optimal and PROP1 for both goods and chores. On the other extreme, it is known that the notion of \emph{proportionality up to the least valued item} (PROPx) may not be achievable even for small instances with three agents~\cite{Moulin2019,FSsurvey,AMS2020}.

The PROP1 and PROPx notions are analogs of relaxations that have been proposed and studied for another very important notion of fairness: \emph{envy-freeness} (EF). An allocation is said to be envy-free if no agent would prefer to be allocated some other agent's bundle over her own. The example with the single indivisible item discussed in the introduction shows that envy-free outcomes may not exist, which motivated the approximate fairness notions of \emph{envy-freeness up to the most valued item} (EF1)~\cite{Budish10} and \emph{envy-freeness up to the least valued item} (EFx)~\cite{CKMPS19}. These two notions permit each agent $i$ some additive amount of envy toward some other agent $j$, but this is at most $i$'s highest value for an item in $j$'s bundle in EF1 and at most $i$'s lowest value for an item in $j$'s bundle in EFx.

The existence of EF1 allocations was implied by an older, and classic, argument by \citet{Lipton}. \citet{CKMPS19} demonstrated that the allocation maximizing the Nash social welfare (the geometric mean of the agents' valuations) satisfies both EF1 and Pareto optimality. But, computing this allocation is APX-hard~\cite{Lee17}, so \citet{BarmanKV18} went a step further by designing a pseudo-polynomial time algorithm that computes an EF1 and Pareto optimal allocation. On the other hand, the progress on the EFx notion has been much more limited. \citet{Plautt2018} proved that EFx allocations always exist in two-agent instances, even for general valuations beyond additive, and a recent breakthrough by \citet{CGM2020} showed that EFx allocations always exist in all instances with three additive agents. Even though this result applies only to instances with three agents, its proof required a very careful and cumbersome case analysis to show how an EFx allocation can be produced for all possible scenarios. Whether an EFx allocation always exists or not for instances of four or more agents is a major open question in fair division.


The \emph{maximin share} (MMS), originally defined by \citet{Budish10}, is an alternative relaxation of proportionality that uses a ``maximin'' argument to define the minimum amount of utility that each agent ``deserves''. However, similarly to PROPx, an allocation satisfying this notion of fairness may not always exist, even for three-agent instances~\cite{KPW2018}. To circumvent this issue, a vibrant line of work has instead aimed to guarantee that every agent's value is always at least $\lambda$ times their MMS benchmark, for some $\lambda\in [0,1]$. The first result along this direction showed that an allocation guaranteeing an approximation of $\lambda= 2/3$ can be computed in polynomial time \cite{AMNS2017}. Subsequent work by \citet{BarmanM17} and \citet{GMT2018} also provided simpler algorithms achieving the same guarantee. \citet{GhodsiHSSY18} then provided a non-polynomial time algorithm producing an allocation guaranteeing $\lambda = 3/4$ and further developed this into a polynomial-time approximation scheme guaranteeing $\lambda =3/4-\epsilon$. The most recent update in this line of work further improved the existence bound to $3/4 + 1/12n$, while also providing a strongly polynomial time algorithm to compute an allocation guaranteeing the $3/4$ approximation~\cite{Garg2020}.

\section{Our Results}
We propose a relaxation of proportionality which we call \emph{proportionality up to the maximin item} ($\propXt$). Just like PROP1 and PROPx, our notion allows the value of each agent $i$ to be less than a $1/n$ fraction of her total value, but by no more than some additive difference $d_i$ which is a function of agent $i$'s value for items allocated to other agents. Rather than going with the most valued item (like PROP1) or the least valued item (like PROPx), our definition of $d_i$ is equal to $\max_{i'\neq i}\min_{j\in X_{i'}} \{v_{ij}\}$, where $X_{i'}$ is the bundle of items allocated to agent $i'$. In other words, we consider the least valued item (from $i$'s perspective) in each of the other agent's bundles, and we take the highest value among them. It is easy to verify that this notion lies between the two extremes of PROP1 and PROPx, and it also captures the maximin element that is used to define the MMS benchmark. To further motivate this notion, in Section~\ref{sec:observations} we show that multiple other natural alternatives fail to exist, even for a single instance with just three agents.

Our main result is a constructive argument proving the existence of a $\propXt$ allocation for any instance with up to five agents. This is in contrast to the PROPx and MMS notions for which existence fails even for three-agent instances. Similarly to the breakthrough by \citet{CGM2020} proving the existence of EFx allocations for three-agent instances, our proof requires a careful case analysis to reach $\propXt$ allocations for each scenario. 

What significantly complicates the arguments for the existence of allocations that satisfy EFx or $\propXt$ is that, according to these notions, the satisfaction of each agent depends not only on what they are allocated but also on how all the remaining items are distributed among the other agents. This leads to non-trivial interdependence which precludes the use of greedy-like algorithms. To streamline our arguments we introduce a notion of \emph{close-to-proportional} bundles, which allow us to decouple the allocation of one subset of agents from another, and reduce the required case analysis. Although we prove the existence for up to five agent instances, this is not due to a hard limit to our approach, other than the fact that the case analysis becomes more complicated and does not provide much more intuition. In fact, we suspect the $\propXt$ property can be satisfied even for instances with an arbitrary number of agents.

\section{Preliminaries}

\newcommand{\agents}{N}
\newcommand{\items}{M}
\newcommand{\allocationverbose}{(X_1, X_2, \dots X_n)}
\newcommand{\minbundle}[2]{m_{#1}(#2)}
\newcommand{\minbundlei}[1]{\minbundle{i}{#1}}
\newcommand{\maxmin}[2]{d_{#1}(#2)}
\newcommand{\maxmini}{\maxmin{i}{X}}

We study the problem of allocating a set $\items$ of $m$ indivisible items (or goods) to a set of $n$ agents $N=\{1,2,\dots,n\}$. Each agent $i$ has a value $v_{ij}\geq 0$ for each good $j$ and her value for receiving some subset of goods $S \subseteq \items$ is additive, i.e., $v_i(S) = \sum_{j \in S}{v_{ij}}$. For ease of presentation, we normalize the valuations so that $v_{i}(M) = 1$  for all $i\in \agents$. Given a bundle of goods $S \subseteq \items$, we let $\minbundlei{S} =\min_{j\in S}\{v_{ij}\}$ denote the least valuable good for agent $i$ in bundle $S$. 

An allocation $X = \allocationverbose$ is a partition of the goods into bundles 
such that $X_i$ is the bundle allocated to agent $i$.
%
%
Given an allocation $X$, we use $\maxmini = \max_{i' \neq i}\{\minbundlei{X_{i'}}\}$ to denote agent $i$'s value for her \emph{maximin good in $X$}, and we say that an agent $i$ is \emph{$\propXt$ satisfied} by $X$ if $v_i(X_i) + \maxmini \geq 1/n$. In turn, an allocation $X$ is $\propXt$ if every agent is $\propXt$ satisfied by it. 

Given a positive integer $k\leq n$ and a set of goods $S\subseteq M$, the \emph{close-to-proportional} (CP) bundle for agent $i$, denoted $\spxi{i}{k}{S}$, is the most valuable subset of goods $B \subset S$ from agent $i$'s perspective for which $v_{i}(B)\leq \frac{1}{k}v_i(S)$. In other words, if $i$ was one of $k$ agents that need to be allocated the set of goods $S$, then $\spxi{i}{k}{S}$ is the most valuable subset of these goods that agent $i$ could receive without exceeding her proportional share. When there are multiple bundles that satisfy this definition, then we let $\spxi{i}{k}{S}$ be one with the maximum cardinality, breaking ties arbitrarily among them. As we discuss in Section~\ref{sec:conclusion}, finding the $\spx$ sets is computationally hard.

\section{Initial Observations}\label{sec:observations}

Before proving some helpful lemmas regarding the $\propXt$ notion and the $\spx$ sets, we first prove that a list of natural alternatives to $\propXt$ fail to exist, even for a simple instance involving just three agents and seven items. Rather than adding the maximin value, $d_i(X)$, to each agent $i$'s value in $X$, we consider adding other alternative functions of the agent's value for the items she did not receive. For example, the value added could be equal to the mean, 
the median, the mode, or the minimax value of agent $i$ for the items in $M\setminus X_i$. 

Consider an instance with seven items and three agents that are identical (with respect to their valuations). One of the items has a high value of 1–6$\epsilon$ for some arbitrarily small constant $\epsilon>0$, and the remaining six items each have a small value $\epsilon$. For any allocation of the items, it is easy to verify that there always exists an agent who did not receive the high value item and also received at most three of the other items; as a result, that agent’s value is at most 3$\epsilon$. It is easy to verify that this agent would violate approximate proportionality for all of the approximate notions proposed above, i.e., the mean (which would add less than 0.25), the median (which would add $\epsilon$), the mode (which would add $\epsilon$), and the minimax item value (which would add $\epsilon$). 

In general, many alternatives to $\propXt$ suffer from the same type of issue: if we introduce dummy items to an instance, i.e., items of insignificant value, the relaxation that these alternative notions provide relative to the exact proportionality vanishes, making them impossible to satisfy in general. Our $\propXt$ notion provides an interesting and non-trivial benchmark that is not susceptible to this issue.

We now proceed to some initial observations regarding the construction of $\propXt$ allocations and $\spx$ sets. Our first observation provides us with a sufficient condition under which ``locally'' satisfying $\propXt$ can lead to a ``globally'' $\propXt$ allocation. Given an allocation of a subset of items to a subset of agents, we say that this partial allocation is $\propXt$ if the agents involved would be $\propXt$ satisfied if no other agents or items were present.

\begin{observation}\label{obs:disjointAlloc}
Let $N_1, N_2$ be two disjoint sets of agents, let $M_1$ and $M_2 = M \setminus M_1$ be a partition of the items into two sets, and let $X$ be an allocation of the items in $M_1$ to agents in $N_1$ and items in $M_2$ to agents in $N_2$. Then, if some agent $i \in N_1$ is $\propXt$ satisfied with respect to the partial allocation of the items in $M_1$ to the agents in $N_1$, and $v_i(M_1) \geq \frac{|N_1|}{|N_1 + N_2|}$, then $i$ is $\propXt$ satisfied by $X$ regardless of how the items in $M_2$ are allocated to agents in $N_2$.  
\end{observation}
\begin{proof}
This follows from the definition of $\propXt$.  
For all $i \in N_1$ we have $d_i(X) \geq \text{max}_{i' \in N_1\setminus\{i\}}\{\text{min}_{j \in X_{i'}}\{v_{ij}\}\}$.  Then, if $v_i(X_i) + \text{max}_{i' \in N_1\setminus\{i\}}\{\text{min}_{j \in X_{i'}}\{v_{ij}\}\} \geq \frac{v_i({M_1})}{|N_1|}$ (i.e., $i$ is $\propXt$ satisfied by $X$ with respect to the agents in $N_1$ and items in $M_1$) and $v_i(M_1) \geq \frac{|N_1|}{|N_1 + N_2|}$, it must be that $v_i(X_i) + d_i(X) \geq \frac{1}{n}$ so $i$ is also $\propXt$ satisfied by $X$ in the overall allocation of the items in $M$ to $N_1 \cup N_2$.
\end{proof}

We now observe that we may, without loss of generality, assume that  $v_{ij} \leq 1/n$ for every agent $i$ and item $j$.  

\begin{lemma}
If there exists some agent $i \in \agents$ and item $j \in \items$ such that $v_{ij} > 1/n$, we may allocate item $j$ to agent $i$ and reduce the problem to finding a $\propXt$ partial allocation of the items in $\items \setminus \{j\}$ to agents in $N\setminus\{i\}$.
\end{lemma}
\begin{proof}
 Let $X$ be an allocation which gives $j$ to agent $i$ and is a $\propXt$ allocation with respect to items in $\items \setminus \{j\}$ and agents in $N\setminus\{i\}$. Observe that agent $i$ is clearly $\propXt$ satisfied by $X$ (she is, in fact, proportionally satisfied).  If any other agent $i' \neq i$ also has value $v_{i'j} > 1/n$ for this item, then $d_{i'}(X)\geq 1/n$ (since $j$ is the only item in $X_i$). This implies that $i'$ is $\propXt$ satisfied since $v_{i'}(X_{i'}) + d_{i'}(X) \geq d_{i'}(X) \geq 1/n$.  Finally, all remaining agents $i'\neq i$ have $v_{i'j} \leq 1/n$ implying that $v_{i'}(\items \setminus \{j\}) \geq \frac{n-1}{n}$ and since $i'$ is $\propXt$ satisfied by $X$ with respect to the items in $\items \setminus \{j\}$ she must be $\propXt$ satisfied with respect to the entire allocation by Observation \ref{obs:disjointAlloc} substituting $N_1 = N \setminus \{i\}$ and $M_1 = M \setminus \{j\}$.
\end{proof}


Our next observation provides some initial intuition regarding why $\spx$ sets play a central role in this paper.

\begin{observation}\label{obs:spxSatisfies}
If agent $i$ is allocated her $\spxi{i}{n}{\items}$ set, then $i$ is guaranteed to be $\propXt$ satisfied  regardless of how the other items are allocated. 
\end{observation}
\begin{proof}
Let $S$ be the $\spxi{i}{n}{\items}$ set of agent $i$ and consider an arbitrary allocation of $M \setminus S$ among the remaining $n-1$ agents.  By definition $v_i(S) + \min_{j \in M \setminus S}{v_{ij}} \geq 1/n$, so it must be that if $i$ is allocated $S$, she is $\propXt$ satisfied.
\end{proof}

We now provide a ``recursive'' construction of $\spxi{i}{k}{S}$ sets which gives us even stronger guarantees. Suppose we ask some agent $i$ to first define the bundle $S_n= \spxi{i}{n}{\items}$, then the bundle $S_{n-1} = \spxi{i}{n-1}{\items \setminus S_n}$, then the bundle $S_{n-2} = \spxi{i}{n-2}{\items \setminus (S_n \cup S_{n-1})}$, and so on.  We show that as long as $i$ receives one of these bundles, then we have some flexibility over how to allocate the remaining items.

\begin{theorem}\label{thm:recSpx}
Let $S_n, S_{n-1}, \dots, S_1$ be the recursively defined $\spx$ sets for some agent $i$, as above. If this agent receives any bundle $S_{\ell}$ and no item from $S_n \cup S_{n-1} \cup \cdots \cup S_{\ell + 1}$ is allocated to the same agent as an item from $S_{\ell - 1} \cup S_{\ell - 2} \cup \cdots \cup S_1$, then agent $i$ will be $\propXt$ satisfied.
\end{theorem}
\begin{proof}
For all $k \in [n]$, we have $v_i(S_k) \leq \frac{1}{k}v_i(M \setminus (S_{n} \cup S_{n-1} \cup \dots \cup S_{k+1}))$ by definition of $S_k$.  Applying this upper bound on $v_i(S_k)$ for $k = n$, because $v_i(M) = 1$ we have that $v_i(M\setminus S_n) \geq 1 - \frac{1}{n} = \frac{n-1}{n}$.  By  applying the upper bound on $v_i(S_k)$ for $k = n-1$ and our lower bound on $v_i(M \setminus S_n)$ we get $v_i(M \setminus (S_n \cup S_{n-1})) \geq \frac{n-1}{n} - \frac{1}{n-1} \cdot \frac{n-1}{n} \geq \frac{n-2}{n}$.  Iteratively repeating this process, we obtain that for all $k \in [n]$ we know that $v_i(M \setminus (S_{n} \cup S_{n-1} \cup \dots \cup S_{k})) \geq \frac{k-1}{n}$.  Also by definition, each $S_{k}$ is a  $\spxi{i}{k}{M \setminus (S_{n} \cup S_{n-1} \cup \dots \cup S_{k+1})}$ set  for  $M \setminus (S_{n} \cup S_{n-1} \cup \dots \cup S_{k+1})$, so we have that $v_i(S_\ell) + \text{min}_{j \in M \setminus (S_{n} \cup S_{n-1} \cup \dots \cup S_{\ell+1})}\{v_{ij}\} \geq \frac{1}{\ell} \cdot v_i(M \setminus (S_{n} \cup S_{n-1} \cup \dots \cup S_{\ell+1})) \geq \frac{1}{\ell} \cdot \frac{\ell}{n} = \frac{1}{n}$.  But finally, as long as the items from $S_{n} \cup S_{n-1} \cup \dots \cup S_{\ell+1}$ are not included in any of the bundles containing the items in $M \setminus (S_{n} \cup S_{n-1} \cup \dots \cup S_{\ell})$ in the complete allocation $X$, we have that $d_i(X) \geq \text{min}_{j \in M \setminus (S_{n} \cup S_{n-1} \cup \dots \cup S_{\ell})}\{v_{ij}\}$ so $i$ is $\propXt$ satisfied when allocated set $S_{\ell}$.
\end{proof}

\section{PROPm Allocations for 4-Agent Instances}

In this section, we demonstrate that $\propXt$ allocations can be found for any instance with 4 agents.  The construction of the allocation proceeds by finding an appropriate initial partition of the items into bundles (based on our notion of $\spx$ bundles) for some arbitrary agent $i$.  Given these bundles, we then show that we have enough freedom in reallocating items to $\propXt$ satisfy each agent.  We note that our proof is constructive, but finding the initial bundles is computationally demanding (as determining if there is some $\spxi{i}{n}{M}$ set with value $1/n$ is an instance of subset sum).

Whenever we say that a set of two or three agents \emph{split} a bundle $\tilde{M}$, we mean that we find a $\propXt$ allocation of the items in $\tilde{M}$ for these agents. Note that \citet{CGM2020} show how to compute EFx allocations for up to three agent instances, and it is easy to verify that EFx outcomes that allocated all the items are also $\propXt$. But, since the arguments for these results are quite complicated and require additional machinery, for completeness in the full version of the paper we provide much simpler arguments for reaching $\propXt$ outcomes with up to three agents using only tools defined herein.

\begin{theorem}
In every instance involving 4 agents with additive valuations there always exists a $\propXt$ allocation.
\end{theorem}
\begin{proof}
We index the agents arbitrarily and begin by recursively constructing $\spx$ sets from the perspective of agent $1$.  We construct 4 bundles of items $A, B, C, D$ as follows:
\begin{itemize}
    \item $C = \spxi{1}{4}{\items}$
    \item $B = \spxi{1}{3}{\items \setminus C}$
    \item $A = \spxi{1}{2}{\items \setminus (C \cup B)}$
    \item $D = \spxi{1}{1}{\items \setminus (C \cup B \cup A)} = M \setminus (A \cup B \cup C)$
\end{itemize}

By Observation \ref{obs:spxSatisfies}, we know that if agent $1$ is allocated bundle $C$, she satisfies $\propXt$.  However, we can also observe that she would be satisfied if she is allocated bundle $D$ because  $v_1(D) \geq 1/4$ (which follows by the repeated application of the definition of $\spx$ sets as in Theorem \ref{thm:recSpx}).

We next want to find bounds on the total value of items in some bundles for agent $1$.  This will allow us to recursively divide the problem into instances with a smaller number of agents.

\begin{lemma}\label{lem:metabundles4}
With agent $1$ and sets $A,B,C,D$ as defined above, $v_1(A \cup D) \geq 1/2$
\end{lemma}

\begin{proof}
By the definition of an $\spx$ set, we have initial upper bounds on the total value agent $1$ has for the generated sets.
\begin{itemize}
    \item $v_1(C) \leq 1/4$
    \item $v_1(B) \leq 1/3(1 - v_1(C))$
    \item $v_1(A) \leq 1/2(1 - v_1(C) - v_1(B))$
\end{itemize}

By combining these upper bounds, we may obtain lower bounds on $v_1(A \cup D)$ as follows 
\begin{align*}
    v_1(A \cup D) &= 1 - (v_1(B) + v_1(C))\\
    &\geq 1 - (1/3 + 2/3(v_1(C)))\\
    & \geq 1 - (1/3 + 1/6 )\\
    & \geq 1/2 ~~~ \qedhere
\end{align*}
\end{proof}



From here we proceed with case analysis based on the value other agents have for $A\cup D$. We present each case as a separate lemma for ease of presentation.

\begin{lemma}
If no agents in $\{2,3,4\}$ have value weakly greater than $1/2$ for the items in $A \cup D$ we can construct an allocation satisfying $\propXt$.
\end{lemma}


\begin{proof}
If there is no agent $i\in \{2,3,4\}$ for which $v_{i}(D) \ge \frac{1}{4}$ then we can give $D$ to agent 1 and split the remaining items between the remaining three agents to produce a $\propXt$ allocation by Observation \ref{obs:disjointAlloc}. Otherwise there must be some agent $i\neq 1$ where $v_{i}(D)\ge \frac{1}{4}$. Then we can give $D$ to agent $i$, give $A$ to agent 1 and split $B\cup C$ between the remaining two agents to arrive at a $\propXt$ allocation by Observation \ref{obs:disjointAlloc} (since for any agent $k$ if $v_{k}(A\cup D)<\frac{1}{2}$, then $v_{k}(B\cup C)\ge \frac{1}{2}$) and Theorem \ref{thm:recSpx}. 
\end{proof}

\begin{lemma}
If one agent in $\{2,3,4\}$ has value weakly greater than $1/2$ for the items in $A \cup D$ we can construct an allocation satisfying $\propXt$.
\end{lemma}

\begin{proof}
Without loss of generality let this be agent 2. Split $A \cup D$ between agents 1 and 2 and split $B \cup C$ between agents 3 and 4 to generate a $\propXt$ allocation by Observation \ref{obs:disjointAlloc}. 
\end{proof}

\begin{lemma}
If exactly two agents in $\{2,3,4\}$ have value weakly greater than $1/2$ for the items in $A \cup D$ we can construct an allocation satisfying $\propXt$.
\end{lemma}


\begin{proof}
Without loss of generality, let agent 2 be the agent who has $v_2(A \cup D) < 1/2$.  For agent 2 it must be that $v_{2}(B)>\frac{1}{4}$ or $v_{2}(C)>\frac{1}{4}$ since $v_{2}(B \cup C)\ge \frac{1}{2}$. But then, we can split $A \cup D$ between the agents 3 and 4, give agent 2 her favorite bundle among $B$ and $C$ and give agent 1 the remaining bundle to arrive at a $\propXt$ allocation by Observation \ref{obs:disjointAlloc} and Theorem \ref{thm:recSpx}. 
\end{proof}

\begin{lemma}
If all three agents in $\{2,3,4\}$ have value weakly greater than $1/2$ for the items in $A \cup D$ we can construct an allocation satisfying $\propXt$.
\end{lemma}

\begin{proof}
If for one of the agents $i \in \{2,3,4\}$ we have that either $v_{i}(B) \ge \frac{1}{4}$ or $v_{i}(C) \ge \frac{1}{4}$ then the allocation follows the same from the previous lemma. Otherwise, we have that all three agents $i\neq 1$ have $v_{i}(C)< \frac{1}{4}$ and we can give $C$ to agent $1$ who is $\propXt$ satisfied by Observation \ref{obs:spxSatisfies} and split the remaining items between the remaining agents which yields a $\propXt$ allocation by Observation \ref{obs:disjointAlloc}. 
\end{proof}

Since in each case, we have demonstrated how one may construct a $\propXt$ allocation, for any set of four agents with additive valuations, a $\propXt$ allocation exists.
\end{proof}

\section{PROPm Allocations for 5-Agent Instances}
In this section, we demonstrate that $\propXt$ allocations can be found for any instance with 5 agents.  The proof proceeds similarly to the four agent case but requires a closer analysis of various cases.  As above, whenever we say that a set of fewer than five agents ``split'' a bundle $\tilde{M}$, we mean that we find a $\propXt$ allocation of the items in $\tilde{M}$ for these agents. 

\begin{theorem}
In every instance involving 5 agents with additive valuations there always exists a $\propXt$ allocation.
\end{theorem}

\begin{proof}
We index the agents arbitrarily and begin by recursively constructing $\spx$ sets from the perspective of agent $1$.  We construct 5 bundles of items $A, B, C, D, E$ as follows:
\begin{itemize}
    \item $D = \spxi{1}{5}{\items}$
    \item $C = \spxi{1}{4}{\items \setminus D}$
    \item $B = \spxi{1}{3}{\items \setminus (C \cup D)}$
    \item $A = \spxi{1}{2}{\items \setminus (B \cup C \cup D)}$
    \item $E = \spxi{1}{1}{M \setminus (A \cup B \cup C \cup D)} = M \setminus (A \cup B \cup C \cup D)$
\end{itemize}

By Observation \ref{obs:spxSatisfies}, we know that if agent $1$ is allocated bundle $D$, she satisfies $\propXt$.  However, we can also observe that she would be satisfied if she is allocated bundle $E$ because  $v_1(E) \geq 1/5$ (which follows by the repeated application of the definition of $\spx$ sets as in Theorem \ref{thm:recSpx}).

We next want to find bounds on the total value of items in some bundles for agent $1$.  This will allow us to recursively divide the problem into instances with a smaller number of agents.

\begin{lemma}\label{lem:metabundles5}
With agent $1$ and sets $A,B,C,D,E$ as defined above, $v_1(A \cup E) \geq 2/5$ and $v_1(A\cup B \cup E) \geq 3/5$.
\end{lemma}

\begin{proof}
By the definition of an $\spx$ set, we have initial upper bounds on the total value agent $1$ has for the generated sets.
\begin{itemize}
    \item $v_1(D) \leq 1/5$
    \item $v_1(C) \leq 1/4(1 - v_1(D))$
    \item $v_1(B) \leq 1/3(1 - v_1(D) - v_1(C))$
    \item $v_1(A) \leq 1/2(1 - v_1(D) - v_1(C) - v_1(B))$
\end{itemize}

By combining these upper bounds, we may obtain lower bounds on $v_1(A \cup E)$ as follows 
\begin{align*}
    v_1(A \cup E) &= 1 - (v_1(B) + v_1(C) + v_1(D))\\
    &\geq 1 - (1/3 + 2/3(v_1(C) + v_1(D)))\\
    & \geq 1 - (1/3 + 1/6 + 1/2v_1(D))\\
    & \geq 1 - (1/2 + 1/10)\\
    & \geq 2/5.
\end{align*}
Similarly, we can lower bound $v_1(A \cup B \cup E)$ as
\begin{align*}
    v_1(A \cup B \cup E) &= 1 - (v_1(C) + v_1(D))\\
    & \geq 1 - (1/4 + 3/4v_1(D))\\
    & \geq 1 - (1/4 + 3/20)\\
    & \geq 3/5 \qedhere
\end{align*}
\end{proof}

With Lemma \ref{lem:metabundles5} in hand, we proceed with case analysis on the value that the other agents have for $A \cup E$ and $A \cup B \cup E$.  We present each case as a separate lemma for ease of presentation.

\begin{lemma}
If all four agents $\{2,3,4,5\}$ have value weakly greater than $3/5$ for the items in $A \cup B \cup E$ we can construct an allocation satisfying $\propXt$.
\end{lemma}
\begin{proof}
Suppose that at least one of the agents $i \in \{2,3,4,5\}$ has $v_i(C) \geq 1/5$ or $v_i(D) \geq 1/5$.  Without loss of generality, let this be agent $2$.  Then, we may give agent $2$ either $C$ or $D$, respectively and $2$ is satisfied.  We can give the other of these two sets to agent $1$ and then then find a $\propXt$ allocation of $A \cup B \cup E$ for agents $\{3,4,5\}$.  By Observation \ref{obs:disjointAlloc} and the assumption that agents $3$, $4$, and $5$ have value at least $3/5$ for $A \cup B \cup E$, we know that they will also be satisfied.  Finally, since we have only repartitioned $A \cup B \cup E$, we know by Theorem \ref{thm:recSpx} that agent $1$ is also satisfied.

Now suppose that all of the agents $i \in \{2,3,4,5\}$ have value $v_i(C) < 1/5$ and $v_i(D) < 1/5$.  Then, by Theorem \ref{thm:recSpx}, we know that we may give $D$ to agent $1$ and reallocate $A \cup B \cup C \cup E$ to the remaining agents and satisfy agent $1$.  But since all four remaining agents have value at least $4/5$ for $A \cup B \cup C \cup E$, by Observation \ref{obs:disjointAlloc} we can then find an allocation $\propXt$ satisfying these agents as well.
\end{proof}

\begin{lemma}
If exactly three of the agents in $\{2,3,4,5\}$ have value weakly greater than $3/5$ for the items in $A \cup B \cup E$ we can construct a $\propXt$ allocation.
\end{lemma}
\begin{proof}
Without loss of generality suppose agent $2$ is the agent who has value $v_2(A \cup B \cup E) < 3/5$.  We can then give agent $2$ her preferred bundle among $C$ and $D$ and agent $1$ the other bundle.  Agent $2$ must be satisfied since she receives value at least $1/5$ and agent $1$ is satisfied regardless of how the items in $A \cup B \cup E$ are distributed by Theorem \ref{thm:recSpx}.  But then, since all $i \in \{3,4,5\}$ have $v_i(A \cup B \cup E) \geq 3/5$ we can split $A \cup B \cup E$ between these agents to obtain a $\propXt$ allocation by Observation \ref{obs:disjointAlloc}.
\end{proof}

\begin{lemma}
If exactly two of the agents in $\{2,3,4,5\}$ have value weakly greater than $3/5$ for the items in $A \cup B \cup E$ we can construct a $\propXt$ allocation.
\end{lemma}
\begin{proof}
Without loss of generality, let agents $4$ and $5$ be the agents with value weakly greater than $3/5$ for the items in $A \cup B \cup E$. By Lemma \ref{lem:metabundles5} we know that agent $1$ also has value greater than $3/5$ for $A \cup B \cup E$.  Further, we know then that agents $2$ and $3$ each have value greater than $2/5$ for the items in $C \cup D$.  By Observation \ref{obs:disjointAlloc}, we can then find a $\propXt$ allocation of these items by splitting $A \cup B \cup E$ between agents $1$, $4$, and $5$ and splitting $C \cup D$ between agents $2$ and $3$.
\end{proof}

We now move to consider the number of agents that have value greater than $2/5$ for $A \cup E$. The case where at most one agent has value at least $3/5$ of $A\cup B\cup E$ is captured in the following lemmas. 

\begin{lemma}
If exactly two of the agents in $\{2,3,4,5\}$ have value weakly greater than $2/5$ for the items in $A \cup E$ we can construct a $\propXt$ allocation.
\end{lemma}
\begin{proof}
Without loss of generality let agents $4$ and $5$ have value weakly greater than $2/5$ for the items in $A \cup E$.  We let these two agents split $A \cup E$ and move to allocate the remaining bundles among agents $1$, $2$, and $3$.  We perform a small case analysis on the number of bundles that agent $2$ or agent $3$ values greater than $1/5$.  

Suppose that agents $2$ and $3$ collectively value at least two distinct bundles in $\{B, C, D\}$ greater than or equal to $1/5$ (i.e., they both value exactly one bundle more than $1/5$ but these bundles are distinct or at least one of the two agents values more than one bundle more than $1/5$).  Then, we may give both of these agents a bundle which they value at least $1/5$ and agent $1$ the remaining bundle to arrive at a $\propXt$ allocation by Observation \ref{obs:disjointAlloc} and Theorem \ref{thm:recSpx}.

Now suppose that agents $2$ and $3$ collectively value exactly one bundle in $\{B, C, D\}$ at least $1/5$.  If this bundle is $B$ or $C$, we know that $v_2(B \cup C) \geq 2/5$ and $v_3(B \cup C) \geq 2/5$ (since $v_2(D) < 1/5$ and $v_3(D) < 1/5$).  We can then allocate $D$ to agent $1$ and split $B \cup C$ between agents $2$ and $3$ to arrive at a $\propXt$ allocation.  If the bundle that $2$ and $3$ value more than $1/5$ is $D$ then we know that $v_2(C \cup D) \geq 2/5$ and $v_3(C \cup D) \geq 2/5$ so we may allocate $B$ to agent $1$ and split $C \cup D$ between agents $2$ and $3$ to arrive at a $\propXt$ allocation by Observation \ref{obs:disjointAlloc} and Theorem \ref{thm:recSpx}.
\end{proof}

\begin{lemma}
If exactly one agent in $\{2,3,4,5\}$ has value weakly greater than $2/5$ for the items in $A \cup E$ we can construct a $\propXt$ allocation.
\end{lemma}
\begin{proof}
Without loss of generality, let agent $5$ have value weakly greater than $2/5$ for the items in $A \cup E$.  By Lemma \ref{lem:metabundles5}, we know that agent $1$ also has value at least $2/5$ for these items, and by assumption it must be that agents $2$, $3$, and $4$ have value at least $3/5$ for the items in $B \cup C \cup D$.  But then, by Observation \ref{obs:disjointAlloc}, we can find a $\propXt$ allocation for all the items by reallocating items in $A \cup E$ to agents $1$ and $5$ and reallocating items in $B \cup C \cup D$ to agents $2$, $3$, and $4$.
\end{proof}

\begin{lemma}
If no agents in $\{2,3,4,5\}$ have value weakly greater than $2/5$ for the items in $A \cup E$ we can construct a $\propXt$ allocation.
\end{lemma}
\begin{proof}
If this is the case, then it must be that all four of these agents have value more than $3/5$ for items in $B \cup C \cup D$.  If none of these agents have value more than $1/5$ for $E$, then we can allocate $E$ to agent $1$ and allocate $A \cup B \cup C \cup D$ to agents $2$, $3$, $4$, and $5$ to arrive at a $\propXt$ allocation.  Suppose, on the other hand, that at least one of these agents, say agent $2$, has $v_2(E) \geq 1/5$, we can allocate $E$ to agent $2$, $A$ to agent $1$ and repartition $B \cup C \cup D$ to agents $3$, $4$, and $5$ to find an allocation that remains $\propXt$ for all agents by Observation \ref{obs:disjointAlloc} and Theorem \ref{thm:recSpx}.
\end{proof}

We now proceed to analyze the four remaining cases which are more elaborate.  
\begin{lemma}
If all four of the agents in $\{2, 3, 4, 5\}$ have value weakly greater than $2/5$ for $A \cup E$ and value less than $3/5$ for $A \cup B \cup E$ we can construct a $\propXt$ allocation.   
\end{lemma}
\begin{proof}
Observe that by assumption all agents $2$, $3$, $4$, and $5$ have value weakly greater than $2/5$ for $C \cup D$.  We can then allocate $B$ to agent $1$ and split $A \cup E$ among agents $2$ and $3$ and $C \cup D$ among $4$ and $5$.  Note that agent $1$ is satisfied by Theorem \ref{thm:recSpx} and since agents $2$ and $3$ split value at least $2/5$ and agents $4$ and $5$ split value at least $2/5$ by Observation \ref{obs:disjointAlloc} we construct a $\propXt$ allocation.
\end{proof}

We then immediately resolve the case when agents $2$, $3$, $4$, and $5$ all have value weakly greater than $2/5$ for $A \cup E$ and exactly one agent, (without loss of generality) say agent $2$,  has value greater than $3/5$ for $A \cup B \cup E$ by following the same allocation described in the previous lemma.
\begin{lemma}
If all four of the agents in $\{2, 3, 4, 5\}$ have value weakly greater than $2/5$ for $A \cup E$ and exactly one of these agents has value weakly greater than $3/5$ for $A \cup B \cup E$ we can construct a $\propXt$ allocation.  
\end{lemma}

The final two cases we examine occur when all but one agent have value at least $2/5$ for $A \cup E$.

\begin{lemma}
If exactly three of the agents in $\{2, 3, 4, 5\}$ have value weakly greater than $2/5$ for $A \cup E$ and all of these agents have value less than $3/5$ for $A \cup B \cup E$ we can construct a $\propXt$ allocation.  
\end{lemma}
\begin{proof}
Without loss of generality, suppose that $v_5(A \cup E) < 2/5$.  Since the remaining agents $i \in \{2,3,4\}$ have $v_i(A \cup E) \geq 2/5$, we can split the set $A \cup E$ between agents $2$ and $3$ and they will be satisfied by Observation \ref{obs:disjointAlloc}.  Moreover, since we have that $v_4(C \cup D) \geq 2/5$ and $v_5(C \cup D) \geq 2/5$ we can split the set $C \cup D$ between agents $4$ and $5$ and they will be satisfied by Observation \ref{obs:disjointAlloc}.  Finally, by assigning $B$ to agent $1$ we construct a $\propXt$ allocation by Theorem \ref{thm:recSpx}.
\end{proof}

\begin{lemma}
If exactly three of the agents in $\{2,3,4,5\}$ have value weakly greater than $2/5$ for $A \cup E$ and exactly one of these agents has value weakly greater than $3/5$ for $A \cup B \cup E$ we can construct a $\propXt$ allocation.
\end{lemma}
\begin{proof}
First suppose that the agent with value less than $2/5$ for $A \cup E$ is the agent with value weakly greater than $3/5$ for $A \cup B \cup E$.  Without loss of generality, let this be agent $5$.  By additivity, it must be that $v_5(B) > 1/5$ so agent $5$ is satisfied by bundle $B$.  We have that $v_3(C \cup D) \geq 2/5$ and $v_4(C \cup D) \geq 2/5$ so we can split the set $C \cup D$ between these agents and they will be satisfied by Observation \ref{obs:disjointAlloc}.  Finally, we know that $v_1(A \cup E) \geq 2/5$ by Lemma \ref{lem:metabundles5} and $v_2(A \cup E) \geq 2/5$ so we may split the set $A \cup E$ between these agents to complete the $\propXt$ allocation by Observation \ref{obs:disjointAlloc}.

On the other hand, suppose that the agent with value less than $2/5$ for $A \cup E$ is not the agent with value weakly greater than $3/5$ for $A \cup B \cup E$.  Without loss of generality, suppose $v_4(A \cup E) < 2/5$ and $v_5(A \cup B \cup E) \geq 3/5$.  We know that $v_3(C \cup D) \geq 2/5$ and $v_4(C \cup D) \geq 2/5$ so we again can split this set between agents $3$ and $4$ and they will be satisfied by Observation \ref{obs:disjointAlloc}.  Since $v_2(A \cup E) \geq 2/5$ and $v_5(A \cup E) \geq 2/5$ we can split $A \cup E$ between $2$ and $5$ and satisfy both by Observation \ref{obs:disjointAlloc}.  Finally, we can give $B$ to agent $1$ to produce a $\propXt$ allocation by Theorem \ref{thm:recSpx}.
\end{proof}

Since in each case, we have demonstrated how one may construct a $\propXt$ allocation, for any set of five agents with additive valuations, a $\propXt$ allocation exists.
\end{proof}

\section{Extensions and Average EFx}

According to our definition, an allocation $X$ is $\propXt$, if for every agent $i$ we have $v_i(X_i) + \maxmini \geq 1/n$, where $\maxmini = \max_{k \neq i}\{\minbundlei{X_k}\}$ is that agent's value for her maximin good in $X$. On the other hand, an allocation $X$ is EFx if for every \emph{pair} of agents $i,k\in N$ we have $v_i(X_i) +m_i(X_k) \geq v_i(X_k)$, where $m_i(X_k)$ is the smallest value of agent $i$ for an item in $X_k$. It is easy to verify that EFx is a stricly more demanding property than $\propXt$. In this section, we propose a middle-ground property between these two extremes, \emph{average-EFx} (a-EFx), which we find to be of interest, and posing a demanding open problem. 

Given some agent $i$, summing up over all $k\in N\setminus\{i\}$ the inequalities that EFx requires for agent $i$, we get:

\begin{align}
\sum_{k\in N\setminus\{i\}} \left(v_i(X_i) +m_i(X_k) \right) &\geq \sum_{k\in N\setminus\{i\}}v_i(X_k)    &\Rightarrow\nonumber\\
    (n-1) v_i(X_i)+\sum_{k\in N\setminus\{i\}} m_i(X_k)  &\geq 1-v_i(X_i) &\Rightarrow \nonumber\\
n v_i(X_i)+\sum_{k\in N\setminus\{i\}} m_i(X_k)  &\geq 1 &\Rightarrow \nonumber\\
v_i(X_i)+\frac{1}{n}\sum_{k\in N\setminus\{i\}} m_i(X_k) &\geq \frac{1}{n}. \label{eq:a-EFx}
\end{align}


We say that an allocation $X$ satisfies a-EFx if Inequality~\eqref{eq:a-EFx} is satisfied for every agent $i\in N$. Clearly, the argument above verifies that EFx implies a-EFx, but the inverse is not true. Specifically, for an agent $i$ to satisfy EFx she needs to not envy any other agent $k$ more than $m_i(X_k)$. On the other hand, agent $i$ could still satisfy a-EFx if she envies some agent $k$ more than $m_i(X_k)$, as long as this extra envy ``vanishes'' after averaging over all agents $k\neq i$, i.e., it satisfies EFx ``on average'', hence the name. Also, note that
\[d_i(X)=\max_{k\in N\setminus\{i\}}\{\minbundlei{X_k}\} \geq \frac{1}{n}\sum_{k\in N\setminus\{i\}} \minbundlei{X_k},\]
so a-EFx implies $\propXt$.
We believe that an interesting open problem is to study the existence of a-EFx allocations in instances with more than 3 agents. Since the $\propXt$ notion is a relaxation of a-EFx, and a-EFx is a relaxation of EFx, this provides an interesting path toward the exciting open problem of whether EFx solutions always exist for instances with 4 or more agents.

\section{Conclusion}\label{sec:conclusion}
Our work defines a new notion of approximate proportionality called $\propXt$. In contrast to similar notions of fairness such as PROPx and MMS, we show that $\propXt$ does exist in the cases of four and five agents with additive valuations. After constructing particular subsets of items for an arbitrary agent (i.e., the close-to-proportional sets), we are able to carefully assign these subsets to agents, or unions of these subsets to a group of agents, and recursively construct $\propXt$ allocations. We conjecture that the existence of $\propXt$ allocations is guaranteed even for instances with more than five agents. The main barrier toward extending our results to these instances seems to be the increasingly complex casework that arises from our approach as the number of agents increases. 

Although we prove the existence of $\propXt$ allocations using a constructive proof, the worst-case running time of our proposed algorithm is not polynomial. In particular, finding a $\spx$ set is at least as hard as subset sum (as one needs to check if some subset gives an agent exactly proportional value), a known NP-hard problem \cite{K1972},  so our approach does not provide an efficient way to calculate a $\propXt$ allocation.  Finding a polynomial time algorithm producing a $\propXt$ allocation for any number of items (and any number of agents) via an alternative method is an interesting possible avenue of future research. Another question we do not explore in this work is achieving $\propXt$ and Pareto efficiency simultaneously.  \citet{AMS2020} provide an algorithm that simultaneously achieves Pareto optimality and PROP1, so an analogous result combining $\propXt$ and Pareto optimality (or proof that no such allocation exists) would nicely complement both their work and ours.

\bibliographystyle{plainnat}
\bibliography{PropXBib}
\newpage
\appendix

\section{PROPm allocations for the 2 and 3 agent cases}\label{sec:3agent}
For completion we include a brief proof of the existence of $\propXt$ allocations in the case of two and three agents.

We note that the proof of the existence of $\propXt$ allocations for two agents is essentially the same as the proof of Theorem 4.3 in \citet{Plautt2018} showing the existence of EFx allocations for two agents.  They employ a ``cut-and-choose'' approach where one agent partitions bundles according to the leximin++ solution and the other agent selects the preferred bundle.  We use the technique except ask the dividing agent $i$ to split the items based on our definition of $\spx$ bundles. 

\begin{theorem}
For $2$ agents with additive valuations one can always find a $\propXt$ allocation.
\end{theorem}
\begin{proof}
Index the two agents arbitrarily and let agent $1$ split $\items$ into two bundles $A$ and $B$ where $A$ is her $\spxi{1}{2}{\items}$ set and $B = \items \setminus A$.  By Observation \ref{obs:spxSatisfies}, if agent $1$ receives $A$ she is satisfied.  Furthermore, by definition $v_1(B) \geq 1/2$, so agent $1$ is satisfied regardless of which bundle she receives.  We can then allow agent $2$ to select her favorite bundle between $A$ and $B$ and by additivity she must obtain value at least $1/2$ (and is therefore $\propXt$ satisfied).
\end{proof}

In the construction of $\propXt$ allocations for $3$ agents when we say that a set of two agents ``split'' a bundle $\tilde{M}$, we mean that we find a $\propXt$ allocation of the items in $\tilde{M}$ for these agents.

\begin{theorem}
For $3$ agents with additive valuations one can always find a $\propXt$ allocation.
\end{theorem}
\begin{proof}
We index the agents arbitrarily and begin by recursively constructing $\spx$ sets from the perspective of agent $1$.  We construct 3 bundles of goods $A, B, C$ as follows:
\begin{itemize}
    \item $B = \spxi{1}{3}{\items}$ 
    \item $A = \spxi{1}{2}{\items \setminus B}$
    \item $C = \spxi{1}{1}{\items \setminus (A \cup B)} = M \setminus (A \cup B)$
\end{itemize}

By Theorem \ref{thm:recSpx}, we know that if agent $1$ is allocated bundle $B$ or $C$ she will be satisfied no matter how the rest is allocated (since $v_1(C) > 1/3$) and if she is allocated $A$ she will be satisfied provided that we can assign bundles $B$ and $C$ to the remaining agents without redistributing items.  We proceed via case analysis on the values agents $2$ and $3$ have for bundles $A$, $B$, and $C$.

\begin{lemma}
If agents $2$ and $3$ collectively value at least two distinct bundles among $A$, $B$, and $C$ greater than or equal to $1/3$, we can construct a $\propXt$ allocation.  
\end{lemma}
\begin{proof}
If agents $2$ and $3$ collectively value at least two distinct bundles greater than $1/3$ then we may assign both agent $2$ and agent $3$ a bundle for which they receive value at least $1/3$.  We can then assign  agent $1$ the remaining bundle to arrive at a $\propXt$ allocation.
\end{proof}

\begin{lemma}
If agents $2$ and $3$ collectively value exactly one bundle among $A$, $B$, and $C$ greater than or equal to $1/3$, we can construct a $\propXt$ allocation.
\end{lemma}
\begin{proof}
Suppose first that this is either bundle $B$ or $A$.  If so, then it must be that $v_2(B \cup A) \geq 2/3$ and $v_3(B \cup A) \geq 2/3$ (since $v_2(C) < 1/3$ and $v_3(C) < 1/3$).  But then we can split $B \cup A$ between agents $2$ and $3$ and allocate $C$ to agent $1$ to obtain a $\propXt$ allocation by Observation \ref{obs:disjointAlloc}.

On the other hand, if this bundle is $C$, then it must be that $v_2(A \cup C) \geq 2/3$ and $v_3(A \cup C) \geq 2/3$ (since $v_2(B) < 1/3$ and $v_3(B) < 1/3$).  But then we can split $A \cup C$ between agents $2$ and $3$ and allocate $B$ to agent $1$ to obtain a $\propXt$ allocation by Observation \ref{obs:disjointAlloc}.
\end{proof}

Since a $\propXt$ allocation can be found in every case, it must always exist for any three agents with additive valuations.
\end{proof}

\end{document}